\documentclass[letterpaper, 10pt, conference]{ieeeconf}
\IEEEoverridecommandlockouts
\renewcommand\footnotemark{}

\makeatletter
	\let\NAT@parse\undefined
\makeatother

\PassOptionsToPackage{noend}{algpseudocode}
\PassOptionsToPackage{nameinlink,capitalize}{cleveref}
\PassOptionsToPackage{dvipsnames}{xcolor}
\let\labelindent\relax
\usepackage[nocompress]{cite}
\usepackage[
	noload={todonotes},
]{myPreamble}

\renewcommand{\ifcaseslineend}{}
\renewcommand{\proofitemizelabeli}{}
\newcommand{\email}[1]{\href{mailto:#1}{#1}}

\crefname{figure}{Figure}{Figures}
\crefname{ALG@line}{step}{steps}
\newcommand{\secref}[1]{\hyperref[#1]{\S\ref*{#1}}}
\newcommand{\figref}[1]{\hyperref[#1]{Fig. \ref*{#1}}}

\usepackage{pifont}% http://ctan.org/pkg/pifont

\usepackage{amscd}
\usepackage{tikz-cd}
\usepackage{curves}
\usepackage{pict2e}

\definecolor{mycolor1}{rgb}{0.00000,0.44700,0.74100}%
\definecolor{mycolor2}{rgb}{0.85000,0.32500,0.09800}%
\definecolor{mycolor3}{rgb}{0.92900,0.69400,0.12500}%
\definecolor{mycolor4}{rgb}{0.49400,0.18400,0.55600}%
\definecolor{mycolor5}{rgb}{0.46600,0.67400,0.18800}%
\definecolor{mycolor6}{rgb}{0.30100,0.74500,0.93300}%
\definecolor{mycolor7}{rgb}{0.63500,0.07800,0.18400}%

\pgfplotsset{
	compat=1.16,
	myaxis/.style = {
		cycle list/GnBu-9,
		xmin = 0    , xmax = 20,
		ymin = 1e-06, ymax = 10,
		xmajorgrids,
		ymajorgrids,
		yminorticks,
		title style  = {font=\bfseries},
		legend style = {%
			rounded corners,
			inner sep = 1pt,
			outer sep = 3pt,
			at = {(0.025,0)},
			anchor = south west,
			legend cell align = left,
			align = left,
			fill opacity = 0.7,
		},
	},
	trajaxis/.style = {
		xmin = -11, xmax = 11,
		ymin = -16, ymax = 22,
		ticks = none,
		cycle list/GnBu-9,
	},
	plotline/.style = {
		line width = 1pt,
		index of colormap = #1 of GnBu-9,
	},
	plotlinemarker/.style = {
		plotline=#1,
		only marks,
		mark repeat=200,
		mark=triangle*,
		mark size = 4pt,
		forget plot,
	},
	plotlinefirstmarker/.style = {
		plotlinemarker = #1,
		mark options = {fill opacity = 0.2, line width = 0.25pt},
		mark phase=1,
	},
	plotlinelastmarker/.style = {
		plotlinemarker=#1,
		mark options = {fill opacity = 1},
		mark phase=101,
	},
}
\usetikzlibrary{intersections,calc,arrows}
\usetikzlibrary{fit,backgrounds,pgfplots.colorbrewer}
\tikzset{%
	mynode/.style = {
		draw,
		circle,
		inner sep=0pt,
		minimum size=15pt,
	},
	arrow>/.style = {%
		->,
		shorten > = 1pt,
	},
	arrow</.style = {%
		<-,
		shorten < = 1pt,
	},
	arrow<>/.style = {%
		<->,
		shorten > = 1pt,
		shorten < = 1pt,
	},
}

\newcommand{\Cone}{{\rm C}}

\let\leq\leqslant
\let\geq\geqslant

\renewcommand\Re{\operatorname{Re}}

\usepackage{caption}
\usepackage{subcaption}

\title{\LARGE\bf
	Algebraic connectivity of layered path graphs under node deletion%
}
\author{\Large
	Ryusei Yoshise\thanks{%
		R. Yoshise is with Kyushu University, Faculty of Mathematics,
		744 Motooka, Nishi-ku, 819-0395 Fukuoka, Japan.
		He is supported by the Graduate Program of Mathematics for Innovation.
		\email{yoshise.ryusei.597@s.kyushu-u.ac.jp}%
	}%
	\texorpdfstring{\and}{ }%
	\Large
	Kaoru Yamamoto\thanks{%
		K. Yamamoto is with Kyushu University, Faculty of Information Science and Electrical Engineering,
		744 Motooka, Nishi-ku, 819-0395 Fukuoka, Japan.
		Her work is supported by JSPS KAKENHI Grants JP19H02161 and JP20K14766.
		\email{yamamoto@ees.kyushu-u.ac.jp}%
	}%
}

\begin{document}

	\maketitle

	\begin{abstract}
		This paper studies the relation between node deletion and algebraic connectivity for graphs with a hierarchical structure represented by layers.
To capture this structure, the concepts of layered path graph and its (sub)graph cone are introduced.
The problem is motivated by a mobile robot formation control guided by a leader.
In particular, we consider a scenario in which robots may leave the network resulting in the removal of the nodes and the associated edges.
We show that the existence of at least one neighbor in the upper layer is crucial for the algebraic connectivity not to deteriorate by node deletion.

	\end{abstract}

	\section{Introduction}
		The algebraic connectivity, i.e., the second smallest eigenvalue of the Laplacian matrix, plays an important role in networked dynamical systems.
It is an indicator of network connectivity \cite{fiedler1973algebraic} and quantifies the convergence speed of consensus algorithms \cite{olfati2004consensus,olfati2007consensus,olshevsky2009convergence}.
It is also an important measure of the robustness of the network to link or node failures \cite{godsil2001algebraic,jamakovic2007relationship}.
For these reasons, in mobile robot networks various control strategies to maximize the algebraic connectivity have been proposed; see, e.g., \cite{kim2005maximizing,degennaro2006decentralized,zavlanos2007potential,zavlanos2008distributed,zavlanos2011graph,satici2013connectivity}.

In this paper, we study a particular type of network structure that consists of multiple layers, each of which is the disjoint union of path graphs as in \cref{fig:layeredpathgraphs}.
The information flow between layers is unidirectional, while it is bidirectional within each layer.
We call such a graph ``layered path graph,'' the formal definition of which is provided in \cref{sec:layeredgraphsdef}.
The problem is motivated by a formation control of mobile robots moving in a flock with hierarchical structure guided by a leader \cite{shen2008cucker,jia2019modelling,nag2022flock}.
We assume that each robot has a communication/sensing range and receives/measures the information of its immediate neighbors in the preceding and the same layers.
In particular, we consider the situation in which a robot (represented by a node) may leave the network, which is captured by the removal of the node and the associated edges.
It is also worth noting that there is a rich and vast literature that relates the spectral properties of a path graph and the control performance in the context of vehicle platooning, such as the $H_{\infty}$ performance against external disturbance \cite{yamamoto2015bounded,pates2018sensitivity,pirani2019graph}, string stability \cite{jovanovic2005ill,herman2014nonzero,zheng2015stability}, and the network resilience and robustness \cite{pirani2017graph,pirani2022impact}, which motivates us to investigate the algebraic connectivity for layered path graphs.

\subsection*{Contribution}
	The main contribution of this work is to provide a set of formal proofs to classify the nodes according to whether their removal deteriorates the network performance with respect to the algebraic connectivity.
	To this end, the concepts of \emph{from-above degree} and \emph{subgraph cone of layers} are introduced that aggregate all the necessary information of upper layers, leading to a particularly simple analysis.

\subsection*{Paper organization}
	The paper is organized as follows.
	\cref{sec:prelim} introduces the terminology and notation used in this paper as well as a prominent result by Tutte \cite{tutte1948dissection}.
	\cref{sec:results} gives the main results about the relation between algebraic connectivity and node deletion with formal proofs.
	In \cref{sec:example} we confirm our results via numerical simulations.
	\cref{sec:conclusions} concludes the paper.

	\begin{figure}[tb]
		\hspace{5.6mm}%
		\includetikz[width=0.35\linewidth]{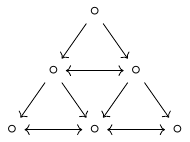}
		\hspace{11.5mm}%
		\includetikz[width=0.35\linewidth]{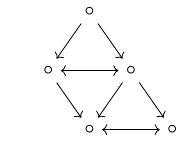}

		\vspace{0.5\baselineskip}%
		\includetikz[width=0.48\linewidth]{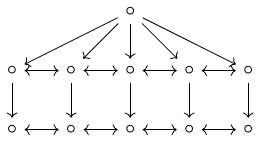}
		\hfill
		\includetikz[width=0.48\linewidth]{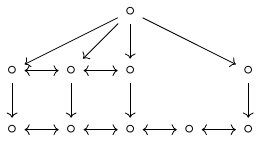}

		\vspace{0.5\baselineskip}%
		\includetikz[width=0.48\linewidth]{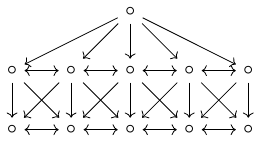}
		\hfill
		\includetikz[width=0.48\linewidth]{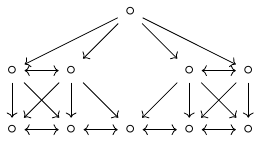}

		\vspace{0.5\baselineskip}%
		\caption{Examples of layered path graphs.}
		\label{fig:layeredpathgraphs}
	\end{figure}

	\section{Preliminaries}\label{sec:prelim}%
		The sets of real numbers and integers are denoted by $\R$ and $\mathbb{Z}$, respectively, and we let $\mathbb{R}_{\geq 0} = \{x \in\mathbb{R}\mid  x \geq 0 \}$ and $\mathbb{Z}_{\geq 0} = \{z \in\mathbb{Z}\mid  z \geq 0 \}$.
% Sets $A$ and $B$ are \emph{disjoint} if $A\cap B =\emptyset$.
The cardinality of a set $A$ is denoted as $\# A$.

For an $n\times n$-matrix $M$, $M[i]$ denotes the $(n-1)\times(n-1)$-matrix obtained from $M$ by removing the $i$-th row and the $i$-th column.
The $(i,j)$ entry of a matrix $M$ is denoted by $[M]_{ij}$.
The eigenvalues of $M$ are denoted as $\{\lambda_i(M)\}_{1\leqslant i \leqslant n}$, where $\{\lambda_i(M)\}_{1\leqslant i \leqslant n}$ satisfy $\Re(\lambda_1(M)) \leqslant \Re(\lambda_2(M)) \leqslant \cdots \leqslant \Re(\lambda_n(M))$.

A \emph{digraph} $G$ consists of a \emph{node set} $V(G)=\{v_1,v_2,\dots,v_n\}$ and an \emph{edge set} $E(G) \subseteq V(G) \times V(G)$, where an edge $e$ is an ordered pair of distinct nodes, i.e., $(v_i, v_j) \in E(G)$ iff there exists a directed edge from $v_i$ to $v_j$ for $v_i \neq v_j$.
We will also use a shorthand notation $e_{ji}$ to denote $(v_i, v_j)$ in the sequel.
%The \emph{weight} on a digraph $G$ is given by a map $w: E(G) \to \mathbb{R}_{\geq 0}$.
A digraph $G$ is the \emph{disjoint union of graphs} $\{G_i\}_{1\leqslant i \leqslant m}$ if $V(G)$ is the disjoint union of $\{V(G_i)\}_{1\leqslant i \leqslant m}$ and $E(G)$ is the disjoint union of $\{E(G_i)\}_{1\leqslant i \leqslant m}$.
A digraph $H$ is an \emph{induced subgraph} of $G$ if every edge of $E(G)$ with both end-nodes in $V(H)$ is in $E(H)$.
A \emph{path graph} is an undirected tree with exactly two leaf nodes.

\begin{comment}
$P_n$ is defined by $V(P_n)=\{v_1, v_2, \ldots , v_n\}$ and $e \in E(P_n)$ iff there exists $1 \leqslant i \leqslant n-1$ s.t. $e=(v_i, v_{i+1})$ or $e=(v_{i+1}, v_i)$.
of subgraphs $\{G_i\}_{1\leq i \leq m}$ if $V(G)=\bigcup_{i} V(G_i)$ and $V(G_i) \cap V(G_j) = \emptyset$ for any $i \neq j$.
\begin{defin}[Disjoint union of subgraphs]%
	A digraph $G$ is a \emph{disjoint union} of subgraphs $\{G_i\}_{1\leq i \leq m}$ if
	$V(G)$ the disjoint union of $V(G_i)$, $i=1,\dots,n$,
	$E(G_i)\subseteq E(G)$ for every $i$, and
	any other edge $e$ belongs to $V(G_i)\times V(G_j)$ for some $i\neq j$.
\end{defin}
A digraph $G$ has \emph{layers} $(G_i)_{1\leq i \leq m}$ if $G$ is the \textcolor{red}{a?} disjoint union of $\{G_i\}$, and for any edge $e \in E(G)$, there exists $i$ s.t. either $e \in V(G_i) \times V(G_i) {\color{red}E(G_i)?}$ or $e \in V(G_i)\times V(G_{i+1})$.

We define that a digraph $G$ is a \emph{layered graph} with \emph{layers} $(G_i)_{1\leqslant i \leqslant m}$ if $V(G)$ is the disjoint union of $\{V(G_i)\}_{1\leqslant i \leqslant m}$, and for any edge $e \in E(G)$, there exists $i$ s.t. either $e \in E(G_i)$ or $e \in V(G_i)\times V(G_{i+1})$.
We define that a layered graph $G$ with layers $(G_i)_{1\leqslant i \leqslant m}$ is a \emph{layered path graph} if it satisfies the following:
\begin{enumerate}
	\item $V(G_1)$ is a singleton,
	\item $G_i$ is the disjoint union of some path graphs.
\end{enumerate}
\end{comment}

The Laplacian matrix of a weighted graph $G^w$, denoted by $L_{G^w}$, is defined as
\[
	[L_{G^w}]_{ij} = \begin{cases}
				\sum_{k\neq i} w(e_{ik}) &\quad\text{if } i=j\\
				-w(e_{ij}) &\quad\text{if } i\neq j
				   \end{cases}
\]
where $w(e)=0$ if $e\not\in E(G^w)$. For an unweighted graph $G$ we simply denote it by $L_G$ with $w(e) = 1$ if $e\in E(G)$.

The Laplacian $L$ of a graph with symmetric weights is always a symmetric positive-semidefinite matrix, with eigenvalues $0=\lambda_1(L) \leq \lambda_2(L) \leq \cdots \leq \lambda_n(L)$.
The second smallest eigenvalue $\lambda_2(L)$, also known as the algebraic connectivity or Fiedler value \cite{fiedler1973algebraic}, is a measure of the graph connectivity and the robustness of the network to node/link failures.

The following theorem is well known in graph theory and will be used in this paper.
\begin{fact}[Tutte's matrix-tree theorem\cite{tutte1948dissection,orlin1978line}]\label{M-tree}
	Let \(G\) be a weighted digraph.
	For any node $v_i \in V(G)$ it holds that
	\[
	\det(L_{G^w}[i])=\sum \left\{ \prod_{e \in T} w(e) \biggm| \text{\begin{tabular}{@{}c@{}}$T$ is a directed spanning\\ tree of $G$ rooted at $v_i$\end{tabular}} \right\},
	\]
	where a directed spanning tree is a spanning tree such that no two directed edges share their tails.
\end{fact}

	\section{Results}\label{sec:results}%
		In this section we provide a set of lemmas and theorems about the algebraic connectivity.
We start from establishing properties for the case where a graph is the disjoint union of any graphs.
Then we move to more specific cases where a graph is equipped with a hierarchical structure represented by \emph{layers}.
The main result is \cref{thm:main} that provides the relation between node deletion and the algebraic connectivity of \emph{layered path graphs}.
To proceed, we first introduce the proposed notions of \emph{layered graphs} and \emph{layered path graphs}.

\subsection{Layered graphs}\label{sec:layeredgraphsdef}
We now equip a graph with a hierarchical structure represented by \emph{layers}.

\begin{defin}[Layered graph]%
A digraph $G$ is a \emph{layered graph} with \emph{layers} $(G_i)_{1\leqslant i \leqslant m}$ if $V(G)$ is the disjoint union of $\{V(G_i)\}_{1\leqslant i \leqslant m}$, and for any edge $e \in E(G)$, there exists $i$ such that either $e \in E(G_i)$ or $e \in V(G_i)\times V(G_{i+1})$.
\end{defin}
\begin{defin}[Layered path graph]%
A \emph{layered path graph} is a layered graph with layers $(G_i)_{1\leqslant i \leqslant m}$ such that
%We say that a layered graph $G$ with layers $(G_i)_{1\leqslant i \leqslant m}$ is a \emph{layered path graph} if it satisfies the following:
\begin{enumerate}
	\item $V(G_1)$ is a singleton, and
	\item each $G_i$ is the disjoint union of path graphs.
\end{enumerate}
\end{defin}

For a node $v$ in a layer $G_i$, we define
\begin{align}
N^{\downarrow}(v) &\coloneqq \{v' \in V(G_{i-1})\mid (v',v)\in E(G)\},
\shortintertext{%
	namely the set of neighbors from the upper layer ($N^{\downarrow}(v) = \emptyset$ for $v \in V(G_1)$), and
}
	d^{\downarrow}(v) &\coloneqq \#N^{\downarrow}(v)
\end{align}
as the \emph{from-above degree}, namely the number of nodes in $N^{\downarrow}(v)$.

We also introduce the notion of the \emph{cone} of a digraph that plays an important role in the subsequent discussions.
Here, we refine and adapt the standard notion of graph cone in algebraic graph theory to our layered structure, thereby introducing integer weights carrying connectivity information.

\begin{defin}[Cone of a digraph]
	For a given map $f: V(G) \to \mathbb{R}_{\geq 0}$, the \emph{cone} of a digraph $G$ is the weighted graph $\Cone G^f$ obtained from $G$ by adding an extra node $v_\ast$, where $V(\Cone G^f) = V(G) \cup \{v_\ast\},\: E(\Cone G^f) = E(G) \cup \{(v_\ast,v) \mid v\in V(G)\},$ and the weight $w: E(\Cone G^f) \to \mathbb{R}_{\geq 0}$ of $\Cone G^f$ is defined by $w(e)=1$ for any $e \in E(G)$, and $w((v_\ast, v)) = f(v)$ for any $v \in V(G)$.
\end{defin}

As shown in \cref{fig:cone}, cones provide convenient representations of specific layers within a layered graph, in which the upper layer is condensed into a single node.
We will later show that it is the from-above degree $d^\downarrow (v)$ that determines whether the removal of a node results in the deterioration in algebraic connectivity, and the cone representation efficiently extracts this information.

\begin{figure}[htb]
	\centering
	\includetikz[width=\linewidth]{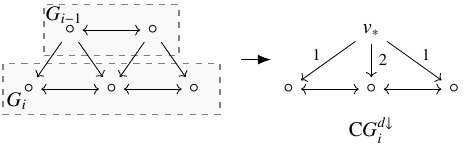}
	\caption{The cone of the \(i\)-th layer $\Cone G_i^{d\downarrow}$ is obtained by condensing the upper layer(s) into a single node and weighing each incoming edge by the from-above degree.}
	\label{fig:cone}
\end{figure}

\subsection{Algebraic connectivity of disjoint union of graphs}
	We first provide the following proposition about the algebraic connectivity of the cone of a digraph which is the disjoint union of any graphs.
	\begin{prop}\label{dis}
		Suppose that a digraph $G$ is the disjoint union of $\{G_i\}_{1 \leqslant i \leqslant m}$. For a given map $f: V(G) \to \R_{\geqslant 0}$,
		\[
			\Re(\lambda_2(L_{\Cone G^f})) = \min \{\Re(\lambda_2(L_{\Cone G_i^{f}})) \mid 1 \leqslant i \leqslant m\}.
		\]
	\end{prop}
	\begin{proof}
		Let $L_i^x := (L_{\Cone G_i^{f}}-xI)[1]$ for $x\in \mathbb{C}$.
		Since $\lambda_1(L_{\Cone G_i^f})=0$, $\lambda_2(L_{\Cone G_i^f})$ is a root of $\det(L_i^x)$.
		The claim now follows by observing that $L_{\Cone G^f}-xI$ can be written as
		\begin{qedequation*}
		\setlength{\arraycolsep}{4pt}%
			L_{\Cone G^f}-xI=
			\begin{bmatrix}
			-x\\[2pt]
			\ast&  L_2^x \\
			\vdots&\ddots&\ddots \\
			\ast&\cdots&\ast& L_{m}^x
			\end{bmatrix}.
		\end{qedequation*}
	\end{proof}

\subsection{Algebraic connectivity of layered graphs}
	We now discuss layered graphs.
	Similarly to \cref{dis}, the algebraic connectivity of a layered graph $G$ can be obtained by considering the cone of each layer for a map $d^{\downarrow}$.
	\begin{thm}\label{layer}
		For a digraph $G$ with layers $(G_i)_{1\leqslant i \leqslant m}$,
		\[
		\Re(\lambda_2(L_G)) = \min \{\Re(\lambda_2(L_{G_1})),  \Re(\lambda_2(L_{\Cone G_i^{d^{\downarrow}}})) \mid 2\leqslant i \leqslant m \}.
		\]
	\end{thm}
	\begin{proof}
		\setlength{\arraycolsep}{3pt}%
		Let $L_i^x := (L_{\Cone G_i^{d^{\downarrow}}}-xI)[1]$ for $x\in \mathbb{C}$.
		Then,
		\[
			L_{\Cone G_i^{d^{\downarrow}}}-xI=
			\left[
			\begin{array}{c|c}
			-x&\\[1pt]
			\hline
			\ast&L_i^x\vphantom{\Big|}
			\end{array}
			\right].
		\]
		Therefore, $\{\lambda_i({L_{\Cone G_i^{d^{\downarrow}}}})\}_{2\leqslant i \leqslant n}$ are the solutions of $\det(L_{\Cone G_i^{d^{\downarrow}}}[1]-xI) = 0$.
		Denoting \(\ell_1^x=L_{G_1}-xI\),
		\[\setlength{\arraycolsep}{4pt}%
			L_G-xI =
			\begin{bmatrix}
			\ell_1^x\\[2pt]
			\ast&  L_2^x \\
			\vdots&\ddots&\ddots \\
			\ast&\cdots&\ast& L_{m}^x
			\end{bmatrix}.
		\]
		Hence,
		\[
			\det(L_G-xI) = \det(\ell_1^x)\det(L_2^x) \det(L_3^x) \cdots \det(L_m^x)
		\]
		and the claim follows.
		%\[\Re(\lambda_2(L_G)) = \min \{\Re(\lambda_2(L_{G_1})),  \Re(\lambda_2(L_{\Cone G_i^{d^{\downarrow}}})) \mid 2\leqslant i \leqslant m \}.\]
	\end{proof}

	\begin{rem}
		We note that if $V(G_1)$ is a singleton, the theorem should be read as
		\begin{qedequation*}
			\Re(\lambda_2(L_G)) = \min \{\Re(\lambda_2(L_{\Cone G_i^{d^{\downarrow}}})) \mid 2\leqslant i \leqslant m \}.
		\end{qedequation*}
	\end{rem}

	Based on \cref{layer} and the classical result, \cref{M-tree}, the following theorem provides the relation between the from-above degree and the lower bound of the algebraic connectivity.

	\begin{thm}\label{alpha}
		Let $G$ be a layered graph with layers $(G_i)_{1\leqslant i \leqslant m}$ where any $G_i$ is an undirected graph.
		If either $G_1$ is a singleton or there exists $\alpha \in \mathbb{R}_{\geqslant 0}$ such that $\lambda_2(L_{G_1}) \geqslant \alpha$, then,
		\[
			d^{\downarrow}(v)\geqslant \alpha~~
			\forall v \in V(G) \setminus V(G_1)
		{}\implies{}
			\lambda_2(L_G) \geqslant \alpha.
		\]
	\end{thm}
	\begin{proof}
		If $\alpha = 0$, the statement is trivial.
		We consider the case of $\alpha > 0$.
		This implies that $\Cone G_i^{d^{\downarrow}}$ contains edges from the extra node $v_\ast$ to all nodes in $G_i$, and hence has a spanning tree with the root $v_\ast$ and the edges $(v_\ast, v)$ for any $v\in V(G_i)$.
		We now consider the weighted graph $\Cone G_i^{f}$ where $V(\Cone G_i^{f}) = V(\Cone G_i^{d^{\downarrow}}),\: E(\Cone G_i^{f}) = E(\Cone G_i^{d^{\downarrow}})$, and the weight $w_x: E(\Cone G_i^f) \rightarrow \R_{\geq 0}$ defined for $x < \alpha$ by $w_x((v_\ast, v))=d^{\downarrow}(v)-x>0\:\forall v\in V(G_i)$ and $w_x(e)=1\forall e \in E(G_i)$.
		Then $\Cone G_i^{f}$ also contains a spanning tree with the root $v_\ast$ and the edges $(v_\ast, v)$ for any $v\in V(G_i)$, and we denote it as $T_1$.
	%	For $x < \alpha$, we define a weight $w_x: E(\Cone G_i^f) \rightarrow \R_{\geq 0}$ by $w_x((v_\ast, v)):=d^{\downarrow}(v)-x>0$ and $w_x(e):=1$ for $e \in E(G_i)$.
		Since $\det(L_{\Cone G_i^{d^{\downarrow}}}-xI)[1] = \det(L_{\Cone G_i^{f}})[1]$, from \cref{M-tree}, we have $\det(L_{\Cone G_i^{d^{\downarrow}}}-xI)[1] = \sum\prod_{e \in T} w_x(e) \geqslant \prod_{e \in T_1} w_x(e) > 0$ for any $x < \alpha$.
		Since the minimum root of $\det(L_{\Cone G_i^{d^{\downarrow}}}-xI)[1]$ equals $\lambda_2(L_{\Cone G_i^{d^{\downarrow}}})$, we find $\lambda_2(L_{\Cone G_i^{d^{\downarrow}}})\geqslant \alpha$.
		By \cref{layer}, $\lambda_2(L_G) \geqslant \alpha$.
	\end{proof}

\subsection{Algebraic connectivity of layered path graphs}
	When each layer is the disjoint union of path graphs, stronger results can be drawn with respect to algebraic connectivity.
	In order to show our main result, \cref{thm:main}, we first provide a set of lemmas and propositions.
	\begin{lem}\label{tridia}
		Let $S=\{k, k+1,  \ldots ,l\} \subsetneq \{1,2,\ldots, n\}$ and $A=T_n(a_1,a_2,\dots,a_n)$ be an $n \times n$ real tridiagonal matrix in the following form:
		\begin{equation}\label{eq:trid}
		\begin{bmatrix}
			a_1&-1\\\
			-1 & \ddots &\ddots\\
			&\ddots&\ddots &-1 \\
			&& -1 &a_n
		\end{bmatrix}.
		\end{equation}
		Then,
		\[\lambda_1(A) <\lambda_1(A_{S})\]
		where $A_{S} := ([A]_{ij})_{\{k \leqslant i, j \leqslant l\}}$.
	\end{lem}

	\begin{proof}
		For $n=2$, we have
		\[
			A= \begin{bmatrix}
			\alpha_1&-1\\
			-1&\alpha_2\\
			\end{bmatrix}
		\]
		and $\lambda_1(A) < \alpha_1, \alpha_2$.

		In the general case, assume that the statement holds for any $n \leqslant m-1$.
		Let $A:=T_m(a_1,a_2,\dots,a_m)$. The Laplace expansion of $\det (A-xI)$ gives
		\[
		\det (A-xI)=(\alpha_{m}-x)\det(A_{S_0}-xI) - \det(A_{S_0'}-xI)
		\]
		where $S_0:=\{1,2,\ldots, m-1\}$ and $S_0':=\{1,2,\ldots,m-2\}$.
		From the assumption,  $\alpha:= \lambda_1(A_{S_0}) < \lambda_1(A_{S_0'})$.
		Therefore, we have $\det (A-\alpha I)=-\det(A_{S_0'}-\alpha I)<0$.
		On the other hand, for some $x \ll 0$, we find $\det (A-x I)>0$.
		Hence, $\lambda_1(A)<\alpha = \lambda_1(A_{S_0})$.
		Similarly, for $S_1=\{2,3,\ldots,m\}$, we find $\lambda_1(A) < \lambda_1(A_{S_1})$.
		If $S$ consists of $m-1$ or less elements, we find $\lambda_1(A)< \lambda_{1}(A_{S_i}) \leqslant \lambda_1(A_S)$ $(i=0,1)$ by the assumption.
	\end{proof}

	\cref{tridia} can be used to show the following lemma:
	\begin{lem}\label{eipath}
		%Let $P_n$ be a path graph and $f: V(P_n) \to \mathbb{Z}_{\geqslant 0}$ be a map satisfying either of the following conditions:
		Let $P$ be the path graph $v_1\leftrightarrow v_2 \leftrightarrow \dots \leftrightarrow v_n$, and let $f: V(P) \to \mathbb{Z}_{\geqslant 0}$ be a map satisfying either of the following conditions:
		\begin{enumerate}
			\item $f(v_1)=0$ or $f(v_n)=0$,
			\item $(f(v_1), f(v_2))=(1,0)$ or $(f(v_{n-1}),f(v_n))=(0,1)$,
			\item $\exists k$ such that $(f(v_k), f(v_{k+1}))=(0,0)$,
			\item $\exists k$ such that $(f(v_k), f(v_{k+1}), f(v_{k+2}))=(1,0,1)$.
		\end{enumerate}
		Then,
		\[\lambda_2(L_{\Cone P^f}) < 1.\]
	\end{lem}

	\begin{proof}
		For notational conciseness, let $f_i:=f(v_i)$.
		Let a matrix $B$ be defined as follows:
		\begin{enumerate}
		\item In case $f_1=0$,
			\[
				B:=
				\begin{bmatrix}
					f_{1}+1&-1\\
					-1&f_{2}+2\\
				\end{bmatrix}.
			\]
		\item In case $(f_1,f_2)=(1,0)$,
			\[
				B:=
				\begin{bmatrix}
					f_{1}+1&-1& \\
					-1&f_{2}+2&-1\\
					&-1&f_{3}+2\\
				\end{bmatrix}.
			\]
		\item In case $(f_k,f_{k+1})=(0,0)$ where $1<k<n-1$,
			\[
				B:=
				\begin{bmatrix}
					f_{k}+2&-1&\\
					-1&f_{k+1}+2&-1\\
					&-1&f_{k+2}+2\\
				\end{bmatrix}.
			\]
		\item In case $(f_k,f_{k+1},f_{k+2})=(1,0,1)$ where $1<k<n-2$,
			\[
				B:=
				\begin{bmatrix}
					f_{k}+2&-1&& \\
					-1&f_{k+1}+2&-1&\\
					&-1&f_{k+2}+2&-1\\
					&&-1&f_{k+3}+2\\
				\end{bmatrix}.
			\]
		\end{enumerate}
		Since $P$, as a path graph, is undirected, notice that the above list covers all the four possibilities in the statement of the lemma, up to possibly reversing the order of the nodes.
		In any of the above cases, we can easily see that $\det(B)>0$ and $\det(B-I)=-1$, implying $\lambda_1(B)<1$.
		Since $B$ can be interpreted as $B=A_S$ for $A := L_{\Cone P_n^{d^\downarrow}}[1]$ and some index set $S$ as in \cref{tridia}, we find $\lambda_2(L_{\Cone P_n^f})=\lambda_1(A) \leqslant \lambda_1(B)<1$.
% 		The cases for $f_n = 0$ and $(f_{n-1},f_{n})=(0,1)$ can be shown similarly.
	\end{proof}

	The following proposition gives the interval in which $\lambda_2(L_G)$ lies for any layered path graph $G$.
	\begin{prop}\label{acinterval}
		For any layered path graph \(G\) it holds that
		\[
		0 \leqslant \lambda_2(L_{G}) \leqslant 1.
		\]
	\end{prop}
	\begin{proof}
		Let $(G_i)_{1\leqslant i \leqslant m}$ be the layers of \(G\).
		By the definition of a layered path graph, $G_2$ is the disjoint union of some path graphs.
		Let \(P\) be one of them.
		Since $V(G_1)$ is a singleton, for any $v \in P$, we have $d^{\downarrow}(v)\in\set{0,1}$.
		Noting that $L_{\Cone P^{d^{\downarrow}}}[1]$ is in the form of \cref{eq:trid}, we find $\lambda_2(L_{\Cone P^{d^{\downarrow}}})=\lambda_1(L_{\Cone P^{d^{\downarrow}}}[1]) \leqslant d^{\downarrow}(v) \leqslant 1$ by \cref{tridia}.
		From \cref{layer} and \cref{dis}, we have $0 \leqslant \lambda_2(L_{G}) \leqslant \lambda_2(L_{\Cone G_2^{d^{\downarrow}}}) \leqslant \lambda_2(L_{\Cone P^{d^{\downarrow}}}) \leqslant 1$.
	\end{proof}

	We are now ready to show the main result.

	\begin{thm}\label{thm:main}
		Let $G$ be a layered path graph with layers $(G_i)_{1\leqslant i \leqslant m}$ satisfying the following condition, \(i=1,\dots,n\):
		\begin{equation}\label{eq:degcondition}
			\forall (v,v') \in E(G_i),\: \#\bigl(N^{\downarrow}(v)\setminus N^{\downarrow}(v')\bigr) \leqslant 1.
		\end{equation}
		Then, the following statements are equivalent:
		\begin{enumerateq}
		\item
			\(0 \leqslant \lambda_2(L_G) < 1\);
		\item
			there exists \(v \in V(G) \setminus V(G_1)\) such that \(d^{\downarrow}(v) = 0\).
		\end{enumerateq}
	\end{thm}
	\begin{proof}
		\begin{proofitemize}
		\item{\it``\(\neg\)(b) \(\Rightarrow\) \(\neg\)(a)''}
			If all nodes \(v \in V(G) \setminus V(G_1)\) satisfy \(d^{\downarrow}(v) \geq 1\), we may invoke \cref{alpha} with \(\alpha=1\) to infer that \(\lambda_2(L_G)\geq1\).
		\item{\it``(b) \(\Rightarrow\) (a)''}
% 			Since $0\leq \lambda_2(L_G) \leq 1$ for any layered path graph by \cref{acinterval}, it is enough to show that
% 			\[
% 				\exists i \geqslant 2 \text{ and } v\in V(G_i) \text{ s.t. } d^{\downarrow}(v) = 0 \implies \lambda_2(L_G) < 1.
% 			\]
			Suppose that \(d^{\downarrow}(v) = 0\) for some \(v\in V(G_i)\) with \(i\geq2\).
			Since $G_i$ is the disjoint union of some path graphs, one of them, be it \(P\), satisfies $v=v_k \in P$.
			If $k=1$ or $k=n$, we can apply case $(i)$ in \cref{eipath} to infer that $\lambda_2(L_{\Cone P^{d^{\downarrow}}})<1$.
			Otherwise, since $\#\bigl(N^{\downarrow}(v_{k-1}) \setminus N^{\downarrow}(v_{k})\bigr) \leqslant 1$ and $\#\bigl(N^{\downarrow}(v_{k+1}) \setminus N^{\downarrow}(v_{k})\bigr) \leqslant 1$, we have
			\(
				(d^{\downarrow}(v_{k-1}), d^{\downarrow}(v_k), d^{\downarrow}(v_{k+1}))\in\{(0,0,0),(0,0,1),(1,0,0),(1,0,1)\}.
			\)
			Hence, we can apply either case $(ii)$, $(iii)$, or $(iv)$ in \cref{eipath} to conclude that $\lambda_2(L_{\Cone P_n^{d^{\downarrow}}})<1$.
			From \cref{layer,dis}, we obtain $\lambda_2(L_{G}) \leqslant \lambda_2(L_{\Cone G_2^{d^{\downarrow}}}) \leqslant \lambda_2(L_{\Cone P_n^{d^{\downarrow}}}) < 1$.
		\qedhere
		\end{proofitemize}
	\end{proof}

	\begin{rem}
	Condition \eqref{eq:degcondition} is naturally satisfied in many common mobile robot formations such as the ones in \cref{fig:layeredpathgraphs}.
	We also note that any induced subgraph $H$ of a layered path graph $G$ with condition \eqref{eq:degcondition} is also a layered path graph possessing this feature, as long as $V(H_1)=V(G_1)$.
	That is, if the original layered path graph satisfies \eqref{eq:degcondition}, then the graphs after removing nodes except the one in $G_1$ also satisfy \eqref{eq:degcondition}.
	Then, \cref{thm:main} together with \cref{acinterval} implies that in any layered path graphs with condition \eqref{eq:degcondition}, if nodes are removed so that all the remaining nodes have at least one edge from the upper layer, the algebraic connectivity stays one. Otherwise, it deteriorates.
	\end{rem}

	\section{Examples}\label{sec:example}%
		In this section we verify our results through numerical examples.
\subsection{Algebraic connectivity}
  We first confirm our main result, \cref{thm:main}, with typical layered path graphs.
	Consider the three layered path graphs $G_{\rm tri}, G_{\rm sq1}, G_{\rm sq2}$ in \cref{fig:Gex}.
	They all satisfy condition \eqref{eq:degcondition} and hence their induced subgraphs are layered path graphs with condition \eqref{eq:degcondition} as long as $v_1$ is not removed.
	We then consider some induced subgraphs of $G_{\rm tri}, G_{\rm sq1}$ and $G_{\rm sq2}$ obtained by removing some nodes from the original graphs.
	The table under each graph in \cref{fig:Gex} shows the relation between the removed nodes and the algebraic connectivity of the induced subgraphs.
	We confirm that, if no node $v$ except $v_1$ satisfies $d^{\downarrow}(v)=0$, the algebraic connectivity equals $1$, i.e., no deterioration.
	However, if there is at least one node satisfying $d^{\downarrow}(v)=0$, it becomes less than $1$.

\begin{figure*}[p]
	\begin{subfigure}[b]{.325\linewidth}
		\begin{minipage}{\linewidth}
			\centering
			\includetikz[width=.75\linewidth]{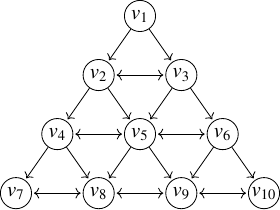}

			\vspace*{5pt}\footnotesize
				\begin{tabular}{|c|c|c|c|}  \hline
	Removed & $\lambda_2$ & $d^{\downarrow}(v)=0$\\ \hline \hline
	\rowcolor{gray!20}%
	-{}-{}- & 1 & -{}-{}- \\ \hline 
	$v_2$ & 0.4679 & $v_4$ \\ \hline 
	$v_4$ & 0.5272 & $v_7$ \\ \hline 
	\rowcolor{gray!20}%
	${v_7}$ & 1 & -{}-{}- \\ \hline 
	$v_2, v_6$ & 0.3820 & $v_4, v_{10}$ \\ \hline 
	$v_4, v_5$ & 0.2360 & $v_7, v_8$ \\ \hline 
	\rowcolor{gray!20}%
	${v_4, v_7}$ & 1 & -{}-{}-  \\ \hline 
	$v_4, v_{10}$ & 0.5188 & $v_7$ \\ \hline
	\rowcolor{gray!20}%
	${v_2, v_4, v_7}$ & 1 & -{}-{}- \\ \hline 
	$v_4, v_5, v_7$ & 0.4679 & $v_8$ \\ \hline 
	$v_4, v_5, v_{10}$ & 0.1981 & $v_7, v_8$ \\ \hline
	$v_2, v_4, v_5, v_7$ & 0.4679 & $v_8$ \\ \hline
	\rowcolor{gray!20}%
	${v_2, v_4, v_7, v_8}$ & 1 & -{}-{}- \\ \hline 
	$v_2, v_4, v_5, v_7, v_{10}$ & 0.3820 & $v_8$ \\ \hline 
\end{tabular}
		\end{minipage}
		\caption{$G_{\rm tri}$}
		\label{fig:Gtri}
	\end{subfigure}
	\hfill
	\begin{subfigure}[b]{.325\linewidth}
		\begin{minipage}{\linewidth}
			\centering
			\includetikz[width=.6\linewidth]{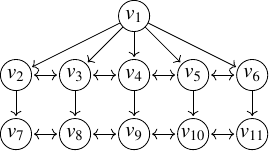}

			\vspace*{5pt}\footnotesize
			\begin{tabular}{|c|c|c|c|}  \hline
	Removed & $\lambda_2$ & $d^{\downarrow}(v)=0$\\ \hline \hline
	\rowcolor{gray!20}%
	-{}-{}- & 1 & -{}-{}- \\ \hline
	$v_2$ & 0.4978 & $v_7$ \\ \hline 
	$v_3$ & 0.6646 & $v_8$ \\ \hline 
	$v_4$ & 0.6972 & $v_9$ \\ \hline 
	$v_2, v_3$ & 0.2434 & $v_7, v_8$ \\ \hline 
	$v_2, v_4$ & 0.4038 & $v_7, v_9$ \\ \hline 
	$v_2, v_5$ & 0.4570 & $v_7, v_{10}$  \\ \hline 
	$v_2, v_6$ & 0.4384 & $v_7, v_{11}$ \\ \hline 
	$v_3, v_4$ & 0.4236 & $v_8, v_9$ \\ \hline 
	$v_3, v_5$ & 0.5188 & $v_8, v_{10}$ \\ \hline 
	$v_2, v_3, v_4$ & 0.1392 & $v_7, v_8, v_9$ \\ \hline 
	$v_2, v_3, v_5$ & 0.2160 & $v_7,  v_8, v_{10}$ \\ \hline
	$v_2, v_3, v_6$ & 0.2278 & $v_7, v_8, v_{11}$ \\ \hline
	$v_2, v_4, v_6$ & 0.3249 & $v_7, v_9, v_{11}$ \\ \hline  
	$v_3, v_4, v_5$ & 0.2679 & $v_8, v_9, v_{10}$ \\ \hline 
	$v_2, v_3, v_4, v_5$ & 0.0810 & $v_7, v_8, v_9, v_{10}$ \\ \hline 
	$v_2, v_3, v_4, v_6$ & 0.1134 & $v_7, v_8, v_9, v_{11}$ \\ \hline 
	$v_2, v_3, v_5, v_6$ & 0.1392 & $v_7, v_8, v_{10}, v_{11}$ \\ \hline 
\end{tabular}
		\end{minipage}
		\caption{$G_{\rm sq1}$}
		\label{fig:Gsq1}
	\end{subfigure}
	\hfill
	\begin{subfigure}[b]{.325\linewidth}
		\begin{minipage}{\linewidth}
			\centering
			\includetikz[width=.6\linewidth]{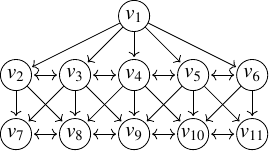}

			\vspace*{5pt}\footnotesize
			\begin{tabular}{|c|c|c|c|}  \hline
	Removed				& $\lambda_2$ & $d^{\downarrow}(v)=0$\\ \hline \hline
%%%%
	\rowcolor{gray!20}%
	-{}-{}- 			&					& 							\\ \cline{1-1}
	\rowcolor{gray!20}%
	$v_2$				&					& 							\\ \cline{1-1}
	\rowcolor{gray!20}%
	$v_3$				&					&							\\ \cline{1-1}
	\rowcolor{gray!20}%
	$v_4$				&\multirow{-4}{*}{1}&\multirow{-4}{*}{-{}-{}-}	\\ \hline
	$v_2, v_3$			& 0.5357 			& $v_7$						\\ \hline
	\rowcolor{gray!20}%
	$v_2, v_4$			&					&							\\ \cline{1-1}
	\rowcolor{gray!20}%
	$v_2, v_5$			&					& 							\\ \cline{1-1}
	\rowcolor{gray!20}%
	$v_2, v_6$			&					& 							\\ \cline{1-1}
	\rowcolor{gray!20}%
	$v_3, v_4$			&					& 							\\ \cline{1-1}
	\rowcolor{gray!20}%
	$v_3, v_5$			&\multirow{-5}{*}{1}&\multirow{-5}{*}{-{}-{}-}	\\ \hline
	$v_2, v_3, v_4$		& 0.2531& $v_7, v_8$	\\ \hline
	$v_2, v_3, v_5$		& 0.5065& $v_7$			\\ \hline
	$v_2, v_3, v_6$		& 0.5337& $v_7$			\\ \hline
	\rowcolor{gray!20}%
	$v_2, v_4, v_6$ 	& 1		& -{}-{}-		\\ \hline
	$v_3, v_4, v_5$		& 0.6972& $v_9$			\\ \hline
	$v_2, v_3, v_4, v_5$& 0.1392& $v_7, v_8,v_9$\\ \hline 
	$v_2, v_3, v_4, v_6$& 0.2434& $v_7, v_8$	\\ \hline 
	$v_2, v_3, v_5, v_6$& 0.4384& $v_7, v_{11}$	\\ \hline 
\end{tabular}
		\end{minipage}
		\caption{$G_{\rm sq2}$}
		\label{fig:Gsq2}
	\end{subfigure}
	\caption{%
		Relation between node deletion and the algebraic connectivity of some induced subgraphs of layered path graphs.
		As shown in \cref{thm:main}, preservation of connectivity is equivalent to the nonvanishing of the from-above degree \(d^{\downarrow}(v)\) for every \(v\in V(G)\setminus V(G_1)\).%
	}
	\label{fig:Gex}
\end{figure*}

\subsection{Formation control of mobile robots}
	We now consider a mobile robot formation control problem where the desired formation is the shape of the graph $G_{\rm tri}$ of \cref{fig:Gtri}.
	The goal here is to achieve a velocity consensus matching that of the leader $v_1$ and target formation.
	The desired formation for the induced subgraph of $G_{\rm tri}$ as a result of node deletion becomes the one with corresponding nodes and edges removed.
% 	For example, when $v_2,v_4,v_7$ are removed, the target formation becomes the one in \cref{fig:targetnoderemove}.
%
% 	\begin{figure}
% 		\centering
% 		\includetikz[width=.5\linewidth]{Triangle_ex2}
% 		\caption{Target formation when $v_2,v_4,v_7$ are removed from $G_{\rm tri}$.}
% 		\label{fig:targetnoderemove}
% 	\end{figure}

	The following standard second-order consensus protocol is considered:
	\[
		\begin{cases}
		\dot{p}_i(t) &{}= v_i\\[4pt]
		\dot{v}_i(t) &{}= -\sum_{j=1}^N [L]_{ij}[(p_i - p_j) + (v_i - v_j)] \quad i = 1,\dots,N
		\end{cases}
	\]
	where $N$ is the number of robots, $p_i \in \R^{2}$ and $v_i\in\R^{2}$ are the position and the velocity of the $i$-th robot in a 2-dimensional space.

	We set $v_1 = (0,1)$ m/s.
	The rest of the robots start from perturbed initial positions and velocities from the desired ones.
	We consider three configurations; $G_{\rm tri}$, $G_{\rm tri}[\{v_2,v_4,v_7\}^C]$, and $G_{\rm tri}[\{v_4,v_5,v_{10}\}^C]$, where $G[\tilde{V}^C]$ denotes the induced subgraph of $G$ obtained by removing nodes in $\tilde{V} \subsetneq V(G)$.
	As we can see from the table in \cref{fig:Gtri}, although three nodes are removed from $G_{\rm tri}$ for both $G_{\rm tri}[\{v_2,v_4,v_7\}^C]$ and $G_{\rm tri}[\{v_4,v_5,v_{10}\}^C]$, the algebraic connectivity of the former remains $1$ while that of the latter is decreased to $0.1981$.
	This fact can be validated from the convergence speed to reach a desired consensus configuration represented in \cref{fig:simulation}.

	\begin{figure*}[p]
		\begin{subfigure}[t]{.325\linewidth}
			\vspace*{0pt}%
			\includetikz[width=\linewidth]{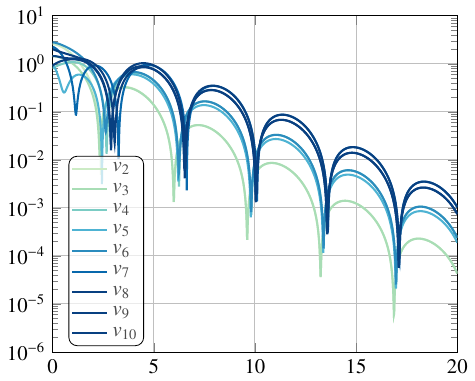}
		\end{subfigure}
		\hfill
		\begin{subfigure}[t]{.325\linewidth}
			\vspace*{0pt}%
			\includetikz[width=\linewidth]{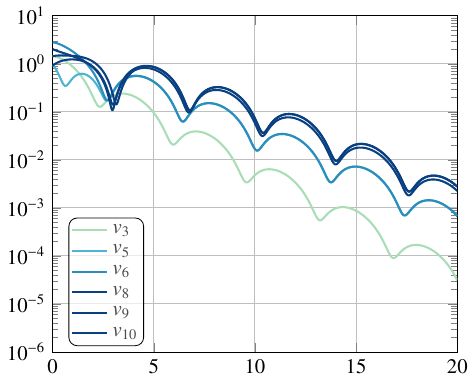}
		\end{subfigure}
		\hfill
		\begin{subfigure}[t]{.325\linewidth}
			\vspace*{0pt}%
			\includetikz[width=\linewidth]{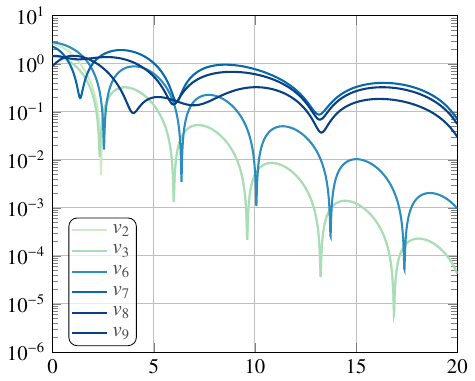}
		\end{subfigure}

		\vspace*{0.5\baselineskip}%
		\begin{subfigure}[t]{.325\linewidth}
			\includetikz[width=\linewidth]{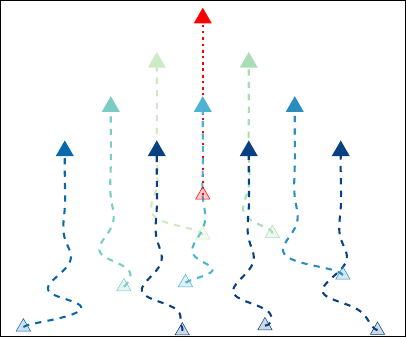}
			\caption{Original graph}
		\end{subfigure}
		\hfill
		\begin{subfigure}[t]{.325\linewidth}
			\includetikz[width=\linewidth]{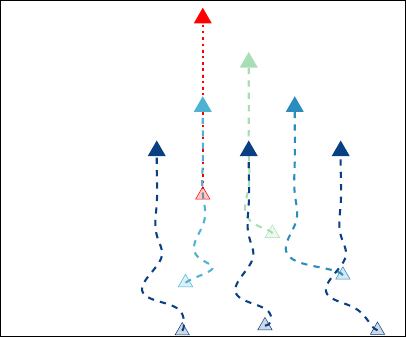}
			\caption{Three nodes ($v_2$, $v_4$ and $v_7$) are removed, yet the connectivity is unaffected}
		\end{subfigure}
		\hfill
		\begin{subfigure}[t]{.325\linewidth}
			\includetikz[width=\linewidth]{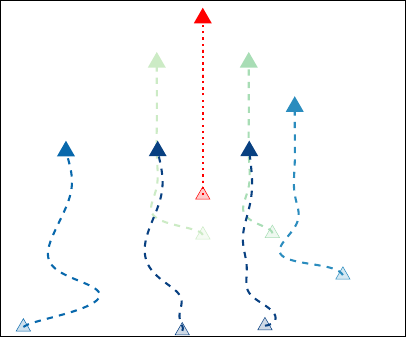}
			\caption{Three nodes ($v_4$, $v_5$ and $v_{10}$) are removed, damaging the connectivity}
		\end{subfigure}
		\caption{%
			Formation control of mobile robots (top row: displacement from desired position over time; bottom row: trajectories).
			The target formations are represented by the graph $G_{\rm tri}$ in {\bf (a)}, and the subgraphs obtained by removing three vertices in {\bf (b)} and {\bf (c)}.
			Despite the loss of three vertices, configuration {\bf (b)} maintains a high connectivity because each follower is still connected to at least one agent from the above layer.
			This is confirmed in the table of \cref{fig:Gtri}, having \(\lambda_2=1\) in this configuration.
			In {\bf (c)}, instead, a slower convergence to a consenus configuration is experienced because of the value \(\lambda_2\approx0.1981\), which occurs because of the loss of connection with the above layer by agents \(v_7\) and \(v_8\).
		}%
		\label{fig:simulation}%
	\end{figure*}

		% \newpage
		% \clearpage
		\section{Conclusions}\label{sec:conclusions}%
		\enlargethispage{-8cm}%
		% \addtolength{\textheight}{-10cm}
		This paper introduced layered path graphs that are common network structures in mobile robot formation control.
The relation between algebraic connectivity and node deletion in such graphs has been studied.
To this end, the concepts of layered path graphs, from-above degree, and subgraph cone are introduced.
In particular, it has been shown that, in order to keep the algebraic connectivity unaffected by node deletion, it is essential to remove nodes so that all remaining nodes receive information from at least one node in the upper layer.

% 	\cleardoublepage

	\bibliographystyle{plain}%
	\bibliography{TeX/Bibliography.bib}%
\end{document}